\documentclass{sig-alternate-per}

\usepackage{amssymb, amsmath}
\usepackage{graphicx}  
\usepackage{subfigure}

\newcommand{\nop}[1]{}

\newtheorem{theorem}{Theorem}
\newtheorem{corollary}{Corollary}
\newtheorem{lemma}{Lemma}

\begin{document}

\title{Delay Bounds for Multiclass FIFO}

\numberofauthors{2}
\author{
\alignauthor
Yuming Jiang \\ 
       \affaddr{Norwegian University of Science and Technology }\\
       \email{jiang@ntnu.no}
\alignauthor
Vishal Misra \\
       \affaddr{Columbia University}\\
       \email{misra@cs.columbia.edu}
}

\maketitle

\begin{abstract}
FIFO is perhaps the simplest scheduling discipline. For single-class FIFO, its delay performance has been extensively studied: The well-known results include an upper bound on the tail of waiting time distribution for $GI/GI/1$ by Kingman \cite{Kingman64} and a deterministic delay bound for $D/D/1$ by Cruz \cite{Cruz91}. However, for multiclass FIFO, few such results are available. To bridge the gap, we prove delay bounds for multiclass FIFO in this work, considering both deterministic and stochastic cases. Specifically, delay bounds are presented for multiclass $D/D/1$, $GI/GI/1$ and $G/G/1$. In addition, examples are provided for several basic settings to demonstrate the obtained bounds in more explicit forms, which are also compared with simulation results. 
\end{abstract}

\section{Introduction}\label{sec-1}

Multiclass FIFO refers to the scheduling discipline where customers are served in the first-in-first-out (FIFO) manner and the services required by\nop{ customers of } different classes may differ. Compared to single-class FIFO where all customers typically have the same service requirements, multiclass FIFO is more general, providing a more natural way to model the system for scenarios where the service requirements of customers from different classes may differ. One example is downlink-sharing in wireless networks, where a wireless base station, shared by multiple users, sends packets to them in the FIFO manner. Since the characteristics of the wireless channel seen by these users may differ, the data rates to them may also be different. Another example is input-queueing on a \nop{router/}switch, where packets are FIFO-queued at the input port before being forwarded to the output ports that pick packets from the FIFO queue and serve at possibly different rates. A third example is video conferencing, where the video part stream and the audio part stream have highly different characteristics. However, the two streams are synchronized when they are generated. In addition, inside the network, they may share the same FIFO queues, e.g. in  network interface cards and in switches. 

Surprisingly, while there are a lot of results\nop{  in the literature} for single-class FIFO, few such results exist for multicalss FIFO. The existing results for multiclass FIFO are mostly under the classic queueing theory (e.g. \cite{ChenZ97}). However, those available results are rather limited and their focus has been mainly on the queue stability condition. In the context of input-queueing in packet-switched systems, multiclass FIFO has also been studied (e.g. \cite{Karol87}). However, in these studies, the focus has been on the throughput of the switch, assuming saturated traffic on each input port. None of these studies has focused on the delay\nop{ bound} performance of multiclass FIFO.  

In terms of delay bounds for FIFO, available results are almost all for single-class FIFO. Among them are the two well-known delay bound results: one by Kingman \cite{Kingman64} for the stochastic case $GI/GI/1$, and one by Cruz \cite{Cruz91} for the deterministic case $D/D/1$. For other or more general arrival and service processes, various delay bounds have also been derived, mainly in the context of network calculus (e.g. \cite{Chang00}\cite{NetCal}\cite{SNetCal}). In particular, for some multiclass settings studied in this work, analytical bounds on their single-class counterparts can be found in the literature (e.g. \cite{Ciucu07b}\cite{Jiang-valuetools09}). With those results, one might expect that they could be readily or easily extended to multi-class FIFO. Unfortunately, such an extension is surprisingly difficult and a direct extension may result in rather limited applicability of the obtained results. 

This work is devoted, as an initial try, to\nop{ filling} bridging the gap. The focus is on finding bounds (or approximations) for the tail distribution of delay\nop{ (in system)} or waiting time\nop{ (in queue)} in multiclass FIFO. 

The rest is organized as follows. In Section \ref{sec-2}, the difficulties in extending or applying single-class FIFO delay bound results to multiclass FIFO are first discussed, using the two well-known delay bounds as examples. Then in Section \ref{sec-3}, we prove delay bounds for multiclass FIFO, considering both deterministic and stochastic cases. Specifically, delay analysis is performed and delay bounds are derived for $D/D/1$, $GI/GI/1$ and $G/G/1$, all under multiclass FIFO. In Section \ref{sec-4}, examples are further provided for several basic settings to demonstrate the obtained bounds in more explicit forms and compare the obtained bounds with simulation results. 

\section{Difficulties} \label{sec-2}

Below we use the Kingman's bound and the Cruz's bound as examples to discuss the difficulties or limitations in applying the single-class\nop{ FIFO} delay bounds to multiclass FIFO. 

\subsection{Kingman's Bound}\label{sec-21}
For a $GI/GI/1$\nop{ FIFO} queue, the following delay bound on waiting time by Kingman \cite{Kingman64} is well-known: 
\begin{equation}\label{Kingman}
P\{W \ge \tau\} \le e^{-\vartheta \tau} 
\end{equation}
where $W$ denotes the steady-state waiting time in queue\nop{ of a customer}, 
$$\vartheta = \sup\{\theta>0: M_{Y(1)-X(1)}(\theta) < 1\}$$ 
with $X(1)$ being the interarrival time and $Y(1)$ the service time of a customer, and $M_{Z}(\theta)$ denotes the moment generating function\nop{ (MGF)} of random variable $Z$, i.e., $M_{Z}(\theta) \equiv E[e^{\theta Z}]$. 

The i.i.d. condition on interarrival times and the i.i.d. condition on service times imply that the bound (\ref{Kingman}) is mostly applicable only for single-class FIFO. Extending it or directly applying it to multiclass FIFO can be difficult. For example, for a multiclass FIFO system, even under that the two i.i.d. conditions hold for each class, the applicability of (\ref{Kingman}) may still be limited due to three main reasons:
\begin{itemize}
\item For the aggregate of all classes, the two i.i.d. conditions for the aggregate may not hold\nop{ even for some basic settings}. For example, for a two-class system with one class being $M/M$ and the other being $D/D$, the i.i.d. conditions, particularly the {\em identical} part, for the aggregate do not hold. 
\item There are cases where, even though within each class, customers are independent, there is dependence between classes. For example, in video conferencing, the video stream and the audio stream are synchronized when generated. As a result, the {\em independent} part of the i.i.d. condition, for the aggregate, may not be met. 
\item For cases where the two i.i.d. conditions for the aggregate also hold, it can still be {\em challenging} to find the probability distribution\nop{ functions or} or characteristic functions of the interarrival times and the service times of the aggregate class of customers, which are needed by (\ref{Kingman}). 
\end{itemize}

We remark that the above reasons also make other related single-class results (e.g. \cite{Ciucu07b}\cite{Jiang-valuetools09}), which rely on similar i.i.d. conditions, difficult to apply to multiclass FIFO. 

\subsection{Cruz's Bound} \label{sec-22}
For a $D/D/1$ \nop{FIFO }system in communication networks, the following delay bound was initially shown by Cruz\nop{ two and half decades ago} \cite{Cruz91}
\begin{equation}\label{Cruz}
D \le \frac{\sigma}{C}
\end{equation}
where $D$ denotes the system delay of any packet, $C$ (in bps) the service rate of the system, and $\sigma$ (in bits) the traffic burstiness parameter. 

The conditions of the Cruz's delay bound are $r \le C$ and that the input traffic during any time interval $[s,s+t]$, denoted as $A(s,s+t)$, for all $s, t\ge 0$, is upper-constrained by 
$A(s, s+t) \le r \cdot t + \sigma$. 

Unlike the Kingman's bound\nop{ (\ref{Kingman})}, the Cruz's bound may be readily used for multiclass FIFO. To illustrate this, consider a FIFO queue with $N$ classes, where the traffic of each class $n$ is upper-constrained by $A_n(s,s+t) \le r_n \cdot t + \sigma_n$ and the service rate of the class is $C_n$ with $r_n \le C_n$. 

Without difficulty, \nop{the Cruz' bound}(\ref{Cruz}) can be extended to this multiclass queue. Specifically, for the aggregate, there holds $A(s, s+t) \equiv \sum_n A_n(s,s+t) \le\sum_n r_n \cdot t + \sum_n \sigma_n$. Then, if there holds 
$$\sum_{n}r_n \le \min_{n} \{C_n\},$$
the delay of any packet is upper-bounded by 
$$D \le \frac{\sum_n \sigma_n}{\min_{n} \{C_n\}}.$$

Unfortunately, the condition $\sum_{n}r_n \le \min_{n} \{C_n\}$ \nop{bound $D \le \frac{\sum_n \sigma_n}{\min_{n} \{C_n\}}$ }can be too restrictive, particularly when the service rates differ much. As a result, the bound can be highly loose or even it cannot be used due to that the condition to use the bound is not met, as to be exemplified later in Section \ref{sec-4}.

\section{Main Results}\label{sec-3}

\subsection{System Model and Notation}
Consider a multiclass\nop{ FIFO} queueing system. There is only one queue that is initially empty.\nop{ at time $0$when there is no customer arrival or service.} Customers are served in the FIFO manner. If multiple customers arrive at the same time, the tie is broken arbitrarily. The size of the queue is unlimited. The serving part of the system is characterized by a work-conserving server. 

There are $N(\ge 1)$ classes of customers (e.g. packets in a communication network). Let $p_n^j$ denote the $j$th customer of class $n$, with $n \in [1,N]$ and $j=1, 2, \cdots$. Each customer $p_n^j$ is characterized by a traffic parameter $l_n^j$ that denotes the amount of traffic (in the number counted on a defined traffic unit, e.g. bits in the communication network setting) carried by the customer. To customers of class $n$, the service rate (in traffic units per second, e.g. bps - bits per second)  of the server is constant, denoted by $C_n$. 

For each class $n$, let $A_n(0,t) \equiv A(t)$ denote the amount of traffic (in traffic units, e.g. bits) that arrives within the time period $[0, t)$, and $A_n(s,t) \equiv A_n(t) - A_n(s)$  the traffic in $[s,t)$. For the aggregate traffic of all classes, $A(s,t) \equiv \sum_n A_n(s,t)$ and $A(t)$\nop{$A(t)=\sum_n A_n(t)$} are similarly defined. Also for each class $n$, we use $\lambda_n$ to denote the average customer arrival rate and $\mu_n$ the average customer service rate, and define $\rho_n = \frac{\lambda_n}{\mu_n}$.

For any customer $p_n^j$, let $a_n^j$ and $d_n^j$ respectively denote its arrival time and departure time. By convention, we let $a_n^{0}= d_n^{0}= 0$. The delay in system of the customer is then $D_n^j = d_n^j-a_n^j$, and the waiting time in queue is $W_n^j = D_n^j - l_n^j / C_n$. 

In addition, corresponding to the notation used in Section \ref{sec-21}, we use $X_n(j)$ to denote the interarrival time between $p_n^{j-1}$ and $p_n^j$, and $Y_n(j)$ the service time of $p_n^j$. By definition, $X_n(j) = a_n^{j} - a_n^{j-1}$, for $j=1, 2, \dots$, and  $Y_n(j) = l_n^j / C_n$. In this paper, we assume that for each class $n$, the processes $X_n(j)$ and $Y_n(j)$ are both stationary. Then by definition, we can also write $\lambda_n = 1/E[X_n(1)]$, $\mu_n = 1/E[Y_n(1)]$, and $\rho_n = E[X_n(1)]^{-1}E[Y_n(1)]$. 

Like single-class FIFO, the dynamics of the multiclass FIFO system are also described by, for all $j = 1, 2, \dots$, 
\begin{equation} \label{fifo-dyn}
d^{j} = \max(a^j, d^{j-1}) + l^j / C^j
\end{equation}
where $j=1, 2, \dots$ denotes the aggregate sequence of all customers ordered according to their arrival times, and $p^j, a^j, d^j, l^j$ and $C^j$ respectively denote the $j$th customer in this ordered aggregate sequence, its arrival time, departure time, carried traffic amount and received service rate. Similarly, we denote the delay of $p^j$ as $D^j = d^j-a^j$, and its waiting time in queue as $W^j = D^j - l^j / C^j$.

\subsection{Delay Bound for Multiclass $D/D/1$} 

Suppose that the traffic of each class $n$ is constrained by 
$
A_n(s, s+t) \le r_n t + \sigma_n
$ 
for all $s, t \ge 0$. For this multiclass $D/D/1$ queue, we have\nop{ the following tight delay bound}: 
\begin{theorem}\label{th-dd1}
If  $\sum_{n}\frac{r_n}{C_n} \le 1$, the delay of any customer $p^j$ is bounded by:
\begin{equation}\label{db-dd1}
D^j \le \sum_{n}\frac{\sigma_n}{C_n}
\end{equation}
and the delay bound is tight.
\end{theorem}

\begin{proof}
For any\nop{ customer} $p^{j}$, there exists some time $t^0$ that starts the busy period where it is in. Note that such a busy period always exists, since in the extreme case, the period is only the service time period of $p^{j}$ and in this case, $t^0=a^{j}$.

Since the system is work-conserving, it is busy with serving customers in $[t^0, d^{j}]$. So, 
$d^{j} = t^0 + \sum_{n=1}^{N}Y_n(t^0, d^{j})$, 
where $Y_n(t^0, d^{j})$ denotes the total service time of class $n$ customers that are served in $[t^0, d^{j}]$. 
Because of FIFO and that the system is empty at $t^0_{-}$, $Y_n(t^0, d^{j})$ is hence limited by the amount of traffic that arrives in $[t^0, a^{j}]$, i.e., $A(t^0, a^{j}_{+})$, where $x_{-} \equiv x - \epsilon$ and $x_{+} \equiv x + \epsilon$ with $\epsilon \to 0$. Specifically, $Y_n(t^0, d^{j}) \le \frac{A_n(t^0, a^{j}_{+})}{C_n} $, 
so we then have 
$d^{j} \le t^0 + \sum_{n=1}^{N}\frac{A_n(t^0, a^{j}_{+})}{C_n}$. 
Under the condition $\sum_{n=1}^{N}\frac{r_n}{C_n} \le 1$, we then obtain:
\begin{eqnarray}
D^{j} &\le& \sum_{n=1}^{N}\frac{A_n(t^0, a^{j}_{+})}{C_n} + t^0-a^{j}  \label{pf1-4} \\ 
&\le& \sum_{n=1}^{N}\frac{r_n\cdot (a^{j}-t^0) + \sigma_n}{C_n} - (a^{j}- t^0) \le \sum_{n=1}^{N} \frac{ \sigma_{n}}{C_n} \nonumber  
\end{eqnarray}
where the second last step is due to the traffic constraint and the last step is from $\sum_{n}\frac{r_n}{C_n} \le 1$ and $a^{j} \ge t^0$. 

Note that, for the system, consider that immidiately after time $0$, every traffic class generates a burst with size $\sigma_n$. In this case, the customer in these bursts, which receives service last, will experience delay $\sum_{n=1}^{N} \frac{ \sigma_n}{C_n}$ that equals the delay bound. So, the bound is tight. 
\end{proof}
 

\subsection{Delay Bounds for Multiclass $GI/GI/1$}

To assist proving results for the ordinary multiclass $GI/GI/1$ system, we first consider a discrete time counterpart of the system, where time is indexed by $t=0, 1, 2, \dots$. The length of unit time is $\delta$. The discrete time system becomes the former by letting $\delta \to 0$. 

In the discrete time system, depending on the length of the unit time $\delta$, it could happen that multiple customers arrive at the same (discrete) time. Because of this, in addition to waiting time\nop{ and (system) delay}, we introduce the concept of virtual waiting time. The virtual waiting time at time $t$ is defined to be the time that a virtual customer, which arrives immediately before time $t$, would experience: {\em All arrivals at $t$ are excluded in the calculation of the virtual waiting time at $t$}. 

More specifically, for any $p^j$, the corresponding virtual waiting time $V^j$ can be written as:
\begin{equation}
V^{j} \equiv \sup_{0 \le s \le a^{j}} \left[ \sum_{n=1}^{N}\frac{A_n(s, a^{j})}{C_n} - (a^{j}-s) \right ]. \label{eq-v} 
\end{equation}
Note that the definition of virtual waiting time also applies to continuous time systems. If in a (continuous time or discrete time) system, there is at most one arrival at a time, then $V^{j}$ equals $W^{j}$, i.e. $W^{j} = V^{j}$.  

For the discrete-time counterpart of multiclass $GI/GI/1$, we have the following result. Its proof is in the Appendix. 

\begin{lemma}\label{th-gdi-d} 
For the discrete time system, if there exists small $\theta > 0$ such that $E[e^{\theta (\sum_{n=1}^{N}\frac{A_n(1)}{C_n} - 1)}] \le 1$, then, for any $p^j$ and for all such $\theta$, 
the virtual waiting time at $a^j$ and the delay of the customer are respectively bounded by,
\begin{eqnarray}
P \{ V^{j} \ge \tau \} &\le& M_{\sum_{n}\frac{A_n(1)}{C_n} - 1}(\theta) e^{-\theta \tau} \label{gd1-i-v}\\
P \{ D^{j} \ge \tau \} &\le& 1- {F}_{\frac{A_1(1)}{C_1}} \ast \cdots \ast  {F}_{\frac{A_N(1)}{C_N}} \ast  {F}_{{V}}(\tau) \label{gdi-i-d}
\end{eqnarray}
where $F_{X}$ denotes the CDF (or a lower bound on CDF) of $X$, ${F}_{{V}}(\tau) \equiv 1 - M_{\sum_{n}\frac{A_n(1)}{C_n} - 1}(\theta) e^{-\theta \tau}$, and  $\ast$ denotes the convolution operation.
\end{lemma}

For Lemma \ref{th-gdi-d}, we highlight that, though in the term ${F}_{\frac{A_1(1)}{C_1}} * \cdots *{F}_{\frac{A_N(1)}{C_N}}$ of (\ref{gdi-i-d}), the time index $1$ is used, the term is actually contributed by concurrent arrivals of the considered $p^j$, specifically by the total service time of all arrivals at time $a^j$, as shown in the proof. 
For the ordinary continuous time multiclass FIFO system, if there is at most one customer arrival at a time, the rest of this part disappears and the remaining is the contribution by the considered customer $p^j$, which is the service time of $p^j$.  Now by letting $\delta \to 0$, the following result immediately follows from Lemma \ref{th-gdi-d}. 

\begin{theorem}\label{th-gigi} 
For a multiclass $GI/GI/1$ system with no concurrent arrivals at any time, if there exists small $\theta > 0$ such that $E[e^{\theta (\sum_{n=1}^{N}\frac{A_n(1)}{C_n} - 1)}] \le 1$, then, 
the waiting time and delay of any customer $p^j$ are respectively bounded by,
\begin{eqnarray}
P \{ W^j \ge \tau \} &\le& e^{-\theta^{*} \tau} \label{gigi1-w}\\ 
P \{ D^j \ge \tau \} &\le& 1- {F}_{Y(j)} * {F}_{{W}}(\tau) \label{gigi1-d}
\end{eqnarray}
where ${F}_{Y(j)}$ is the CDF (or a lower bound on CDF) of the service time of the customer, ${F}_{{W}} = 1- e^{-\theta^{*} \tau}$ and 
$$\theta^{*} = \sup\{\theta>0: E[e^{\theta (\sum_{n}\frac{A_n(1)}{C_n} - 1)}] \le1\}.$$ 
\end{theorem}

\subsection{Delay Bounds for Multiclass $G/G/1$}

As discussed in Section \ref{sec-21}, when there are multiple classes\nop{ of customers}, even though the i.i.d. conditions may hold for each class, such i.i.d. conditions do not necessarily carry over to between classes. As a consequence, Theorem \ref{th-gigi} may not be applicable. To deal with this, we present the following bounds. 

\begin{theorem}\label{th-gg1} 
Suppose the traffic of each class has generalized Stochastically Bounded Burstiness (gSBB) \cite{Yin02} \cite{Jiang-comnet09}, satisfying for some $R_n >0$ and $\forall t >0$, 
$$P \{ \sup_{0 \le s \le t}[A_n(s, t) - R_n \cdot (t-s)] > \sigma \} \le \bar{F}_n(R_n, \sigma)$$
for all $\sigma \ge 0$.  
Then, for $\forall (R_1, \dots, R_N)$, under the condition 
$$\sum_{n}\frac{R_n}{C_n} \le 1,$$ 
the delay of any customer $p^j$ is bounded by (a.s.):
$$
P \{ D^j > \tau \} \le \inf_{p_1 + \cdots + p_N=1}\sum_{n=1}^N \bar{F}_n(R_n, p_n \cdot C_n \tau)
$$
and if the arrival processes of the $N$ classes are independent of each other, then the delay is bounded by (a.s.): 
$$
P \{ D^j > \tau \} \le 1- {F}_1 * \cdots *{F}_N(R_n, C_n\tau) 
$$
where ${F}_n(R_n, C_n \tau) \equiv 1 - \bar{F}_n(R_n, C_n \tau)$. 
\end{theorem}

\begin{proof}
Following the proof of Theorem~\ref{th-dd1}, we have obtained (\ref{pf1-4}). By applying\nop{ the condition} $\sum_{n=1}^{N}\frac{R_n}{C_n} \le 1$ to (\ref{pf1-4}), we further obtain
\begin{eqnarray}
D^{j} 
&\le & \sum_{n=1}^{N} \frac{A_n(t^0, a^{j}_{+}) - R_n\cdot(a^{j}-t^0)}{C_n}.  \label{eq-db-3a}
\end{eqnarray}

Note that in (\ref{eq-db-3a}), $t^0$ is a random variable. Taking all sample paths into consideration, with $A_n(t^0, a^{j}_{+}) - R_n\cdot(a^{j}-t^0) \le \sup_{0 \le s \le a^{j}}[A_n(s, a^{j}_{+}) - R_n\cdot(a^{j}-s)]$ , we get: 
\begin{eqnarray}\label{eq-2}
D^{j} 
&\le& \sum_{n=1}^{N} \frac{\sup_{0 \le s \le a^{j}}[A_n(s, a^{j}_{+}) - R_n\cdot(a^{j}-s)]}{C_n}. \label{eq-db-4b}
\end{eqnarray}

Since the traffic of each class has gSBB, with simple manipulation on the definition and applying $\epsilon \to 0$, we have,
\begin{eqnarray}
& & P \{ \frac{\sup_{0 \le s \le a^{j}}[A_n(s, a^{j}_{+}) - R_n\cdot(a^{j}_{+}-s)]+R_n \epsilon}{C_n}  > \tau \} \nonumber \\
&\le& \bar{F}_n (R_n, C_n \tau). \nonumber
\end{eqnarray}
The theorem then follows from probability theory results on sum of random variables. 
\end{proof}

We remark that it is easily verified that the deterministic traffic model is a special case of the gSBB model with $\bar{F}(x) = 0$ for all $x \ge \sigma_n$ and $\bar{F}(x) = 1$ otherwise. In addition, a wide range of traffic processes have been proved to have gSBB~\cite{Jiang-comnet09}.

\section{Examples} \label{sec-4}

To illustrate\nop{ the use of} the obtained delay bounds, this section presents examples for some basic settings, whose single-class counterparts have been extensively studied in the literature, with more explicit expressions for the delay bounds. 
In addition, the obtained bounds are compared with simulation results and discussed. 

\subsection{Multiclass $D/D/1$} 

Consider a multiclass FIFO queue in a communication network. Assume that there are two traffic classes, and for each class, packets have constant size $l_n$, which arrive periodically with $X_n$ being the period length. 

It is easily verified that for each class, its traffic arrival process satisfies 
$A_n(s,t) \le r_n \cdot (t-s) + \sigma_n$, 
with $r_n = l_n / X_n$ and $\sigma_n = l_n$. Applying them to Theorem \ref{th-dd1}, a delay bound is found as: If $\frac{r_1}{C_1}+\frac{r_2}{C_2} \le 1$, the delay of any packet satisfies: 
\begin{equation}\label{dd1-mc}
D \le \frac{l_1}{C_1} + \frac{l_2}{C2}.
\end{equation}

To compare, recall the delay bound directly from single class $D/D/1$ network calculus analysis discussed in Section \ref{sec-22}, which is, if $r_1+r_2 \le \min\{C_1, C_2\} $, 
\begin{equation}\label{dd1-direc}
D \le \frac{l_1+l_2}{\min\{C_1,C2\}}.
\end{equation}

To illustrate the two bounds, Figure \ref{fig:dd1} is presented, where two cases, Case 1 and Case 2, are considered\nop{ and the bounds are compared with simulation results}. For Case 1, $C_1 =20Mbps$ and $C_2 = 100Mbps$; $l_1=100$ bytes and $l_2=1250$ bytes; $X_1=0.1 ms$ and $X_2=1 ms$. For Case 2, the other settings are the same except that $C_1 = 10Mbps$. 

\begin{figure}[ht!]
    \centering
    \subfigure[Case 1]{
        \includegraphics[width=0.465\linewidth]{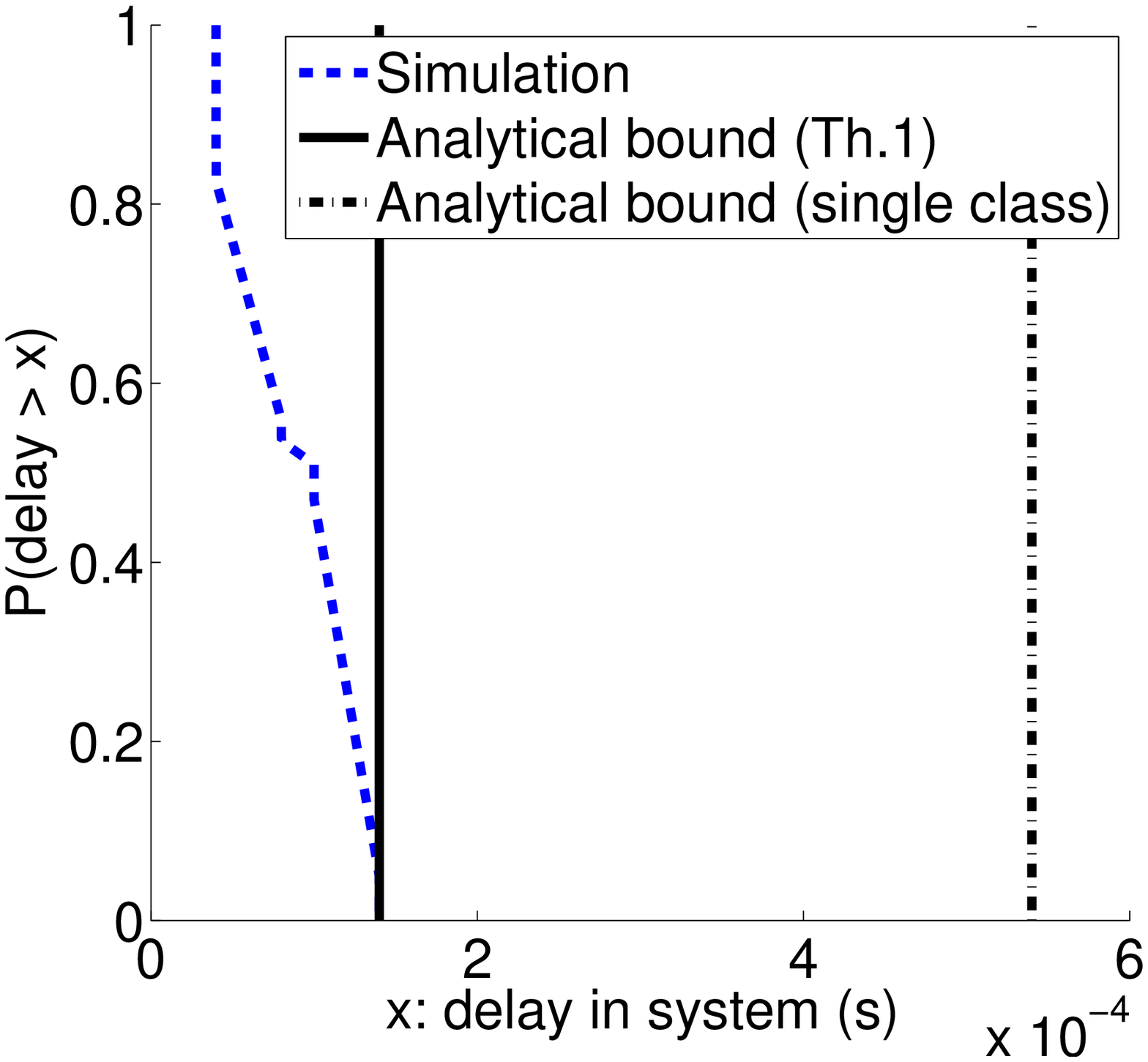} 
        \label{fig:case1}
    }
    \hfill
    \subfigure[Case 2] {
        \includegraphics[width=0.465\linewidth]{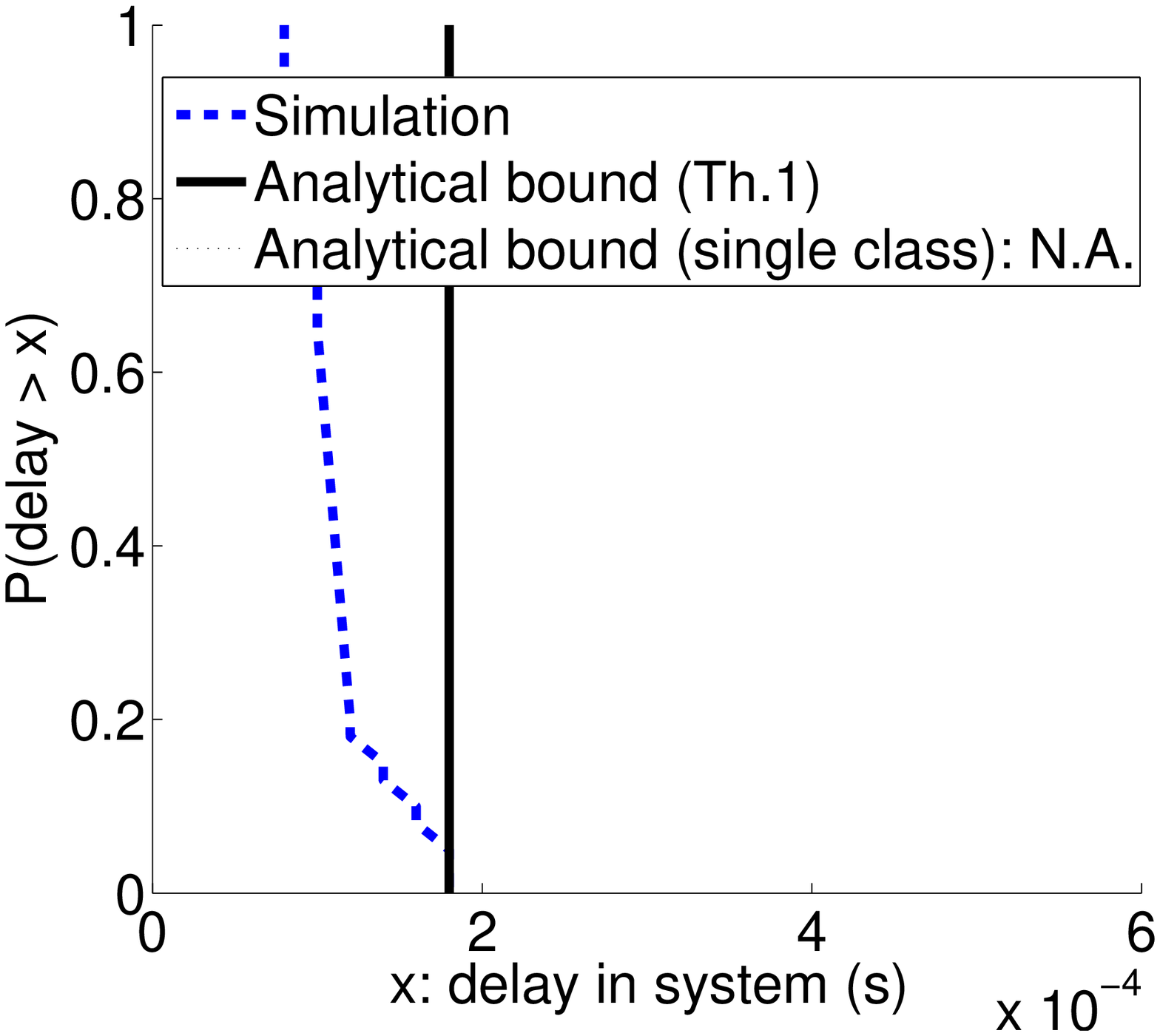} 
        \label{fig:case2}
   }
    \caption{Multiclass $D/D/1$}
    \label{fig:dd1}
\end{figure}

For both cases, $r_1 = 8Mbps$ and $r_2=10Mbps$. It is easily verified that while for both cases, the condition $\frac{r_1}{C_1}+\frac{r_2}{C_2} \le 1$ is satisfied, the condition $r_1+r_2 \le \min\{C_1, C_2\} $ is only met for Case 1. 

Figure \ref{fig:dd1} shows that, for both cases, the bound by Theorem \ref{th-dd1} is not only able to bound the delays of all simulated packets but also tight, i.e., some packets can indeed experience delay equal to the bound\nop{ (\ref{dd1-mc})}. However, \nop{the bound }(\ref{dd1-direc}) can be highly conservative as shown by Figure \ref{fig:case1} for Case 1, and may even not be applicable (N.A.) due to that its required condition is not met as indicated by Figure \ref{fig:case2} for Case 2.

\subsection{Multiclass $GI/GI/1$}

In this subsection, we give two examples for multiclass $GI/GI/1$, which are multiclass $M/D/1$ and multiclass $M/M/1$.  In both examples, customers of each class arrive according to a Poisson process with average interarrival time $X_n$\nop{ or equivalently average arrival rate $\lambda_n \equiv X_n^{-1}$}. In addition, customers of the same class have the same expected service time $Y_n$. However, while for $M/M/1$, the service time of each customer is exponentially distributed, it is constant for $M/D/1$. 

Similar to the $D/D/1$ example, for each of them, we try to give an expression for the tail of delay / waiting time in a closed-form format to help the use of the related bounds. In particular, we have the following corollaries, for which we assume the stability condition is met, i.e. $\rho\equiv\sum_{n}\frac{r_n}{C_n} < 1$. Their proofs are included in the Appendix. 

\begin{corollary}\label{cor:md1}
For multiclass $M/D/1$, if all classes are independent, then for any customer $p_n^j$, its waiting time satisfies: 
$P \{ W_n^{j} >  \tau \} \le e^{-\theta^* \tau}$ 
with
\begin{equation}
\theta^{*} = \sup\{\theta>0:  \sum_n \lambda_n (e^{\theta Y_n}-1)-\theta  \le 0\}
\end{equation} 
an approximation of which is
\begin{equation}
\theta^* = 2(1-\rho)\tau/(\sum_{n=1}^{N} X_n^{-1} Y_n^2). \label{md1-appr}
\end{equation} 
\end{corollary}

\begin{corollary}\label{cor:mm1}
For multiclass $M/M/1$, if all classes are independent, then for any customer $p_n^j$, its waiting time satisfies
$P \{ W_n^{j} > \tau \} \le e^{-\theta^* \tau}$
with
\begin{equation}
\theta^{*} = \sup\{0 < \theta < \min_n \mu_n: \sum_n\frac{\lambda_n}{\mu_n - \theta} \le 1\} \label{mm1-exact}
\end{equation} 
an approximation of which is 
\begin{equation}
\theta^* = (1-\rho)\tau/(\sum_{n=1}^{N} X_n^{-1} Y_n^2). \label{mm1-appr}
\end{equation} 
\end{corollary}

\begin{figure}[ht!]
    \centering
    \subfigure[Case 3: $M/D/1$]{
        \includegraphics[width=0.465\linewidth]{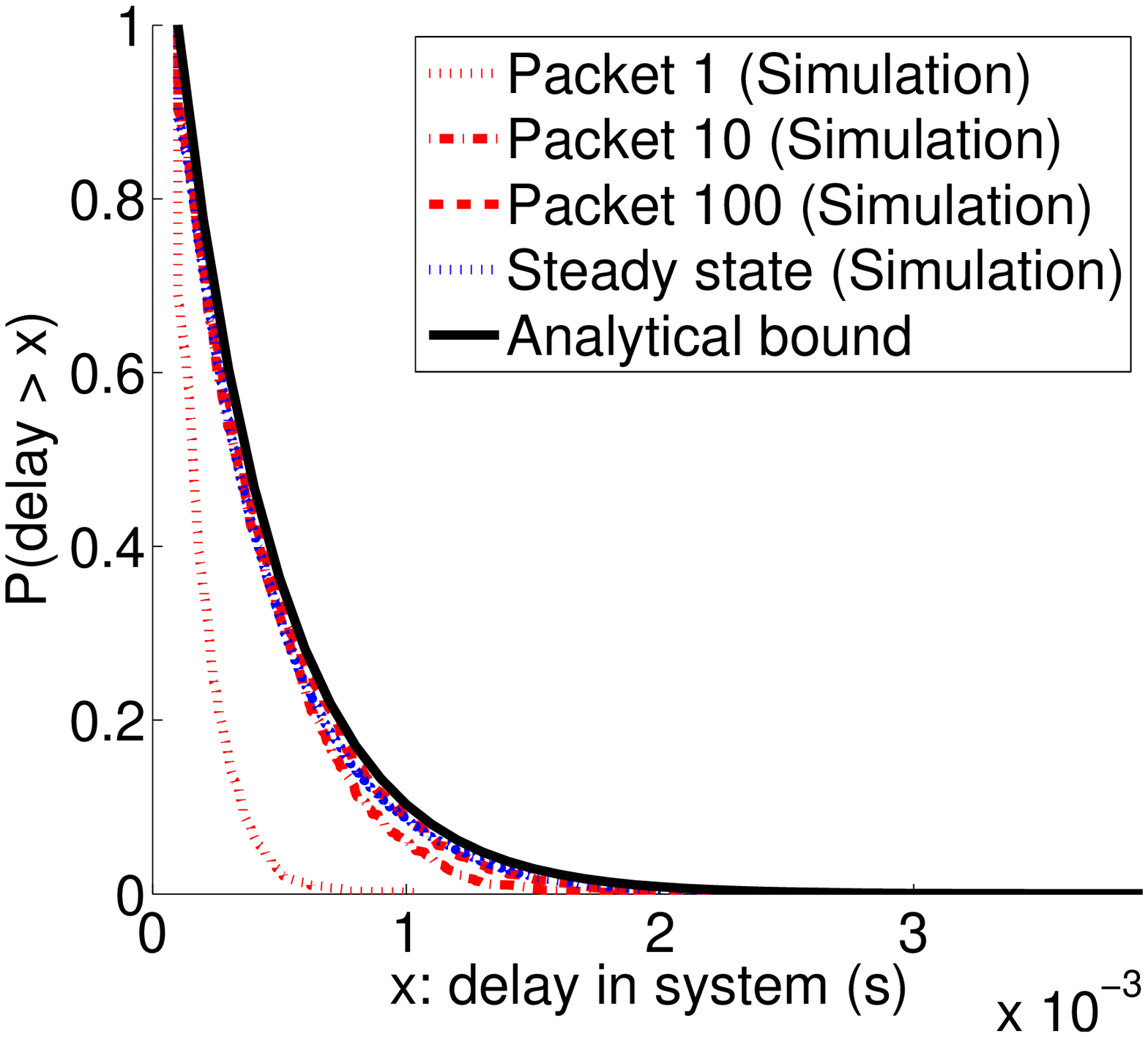} 
        \label{fig:case3}
    }
    \hfill
    \subfigure[Case 4: $M/M/1$] {
        \includegraphics[width=0.465\linewidth]{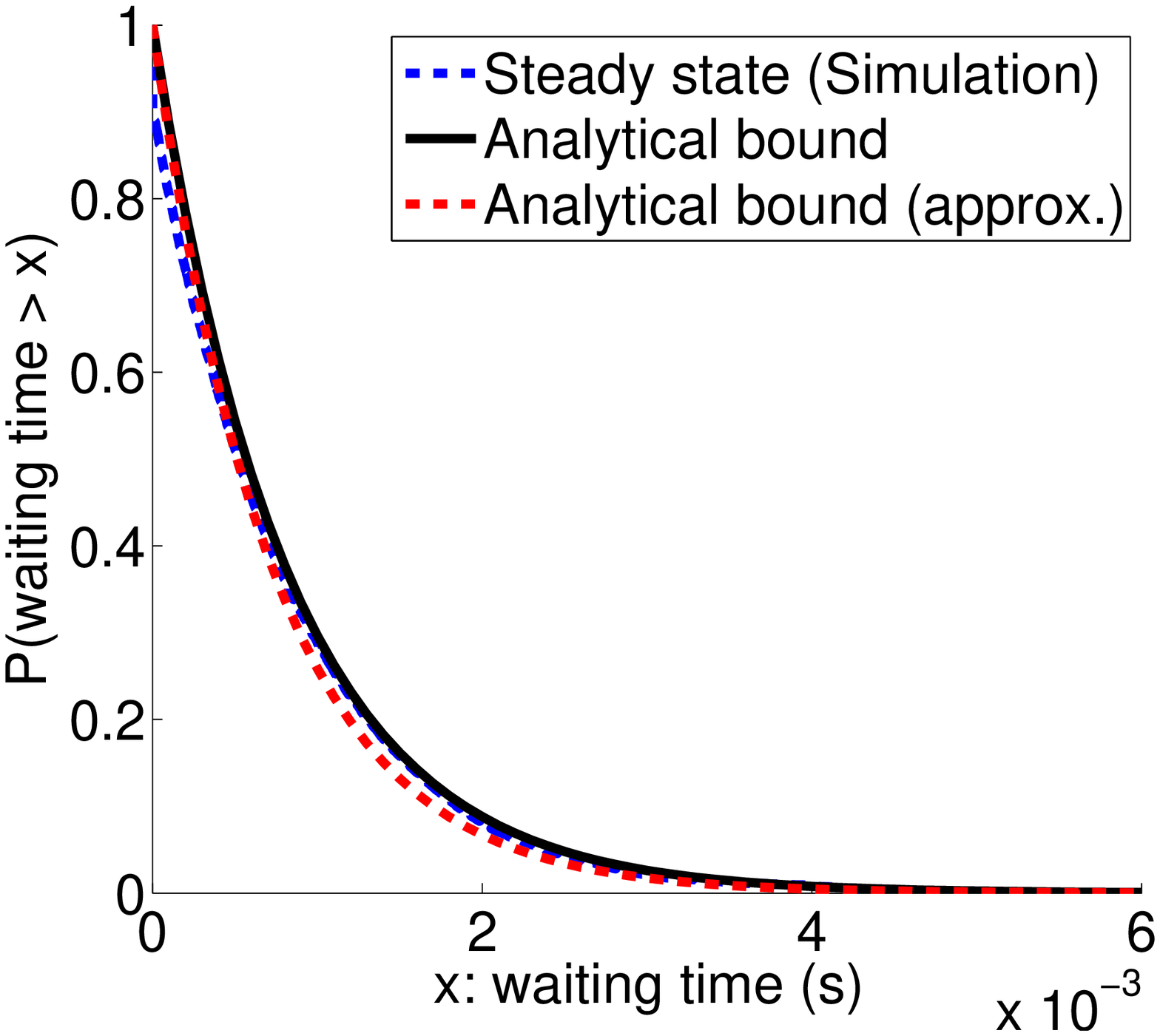} 
        \label{fig:case4}
   }
    \caption{Multiclass $GI/GI/1$}
    \label{fig:gi1}
\end{figure}

We remark that, with the bound on waiting time, a bound on delay can be easily obtained, e.g. for $M/D/1$, $P \{ D_n^{j} >  Y_n + \tau \} =P \{ W_n^{j} >  \tau \}$.
In addition, we remark that the single class version of the bounds in the above two corollaries resemble closely with the literature approximations for tail of delay / waiting time distribution in single class $M/G/1$, e.g., (2.9), (2.124) and (2.143) in \cite{Kleinrock76}.

To illustrate the bounds\nop{ by (\ref{md1-appr}) and  (\ref{mm1-appr})}, Figure \ref{fig:gi1} is presented, where two cases, Case 3 and Case 4, are respectively considered\nop{ and the bounds are compared with simulation results}. For Case 3, the other settings are the same as for Case 2, except that packets of each class arrive according to an independent Poisson process\nop{ with $X_n$ denoting the average interarrival time}, i.e. Case 3 is multiclass $M/D/1$. For Case 4, the other settings are the same as for Case 3, except that the traffic of each packet results in an exponentially distributed service time, i.e. Case 4 is multiclass $M/M/1$. 

In Figure \ref{fig:case3}, the delay CCDF simulation results of the 1st, the 10th and the 100th packet of Class 1, are included, in addition to the steady state delay CCDF and the analytical bound. Figure \ref{fig:case3} shows that the delays of packets are stochastically increasing as the packet number goes higher and they converge to the steady state distribution. This is as expected and it is a proven phenomenon for single-class FIFO (e.g. \cite{Kleinrock76}). For this reason, in later figures, only steady-state delay or waiting time distribution will be focused. In Figure \ref{fig:case4}, the curves are for Class 2, which include the simulated steady-state waiting time CCDF, the analytical bound based on (\ref{mm1-exact})\nop{ the root of (\ref{mm1-exact})} and the approximate analytical bound (\ref{mm1-appr}). 

Figure \ref{fig:gi1} shows that the analytical bounds are fairly tight and provide good approximations of the corresponding steady-state delay CCDF for both cases.

\subsection{Multiclass $G/G/1$}

In this subsection, we consider two examples where, even though within each class, the interarrival times and the service times are still respectively i.i.d., the overall FIFO system no more has i.i.d. interarrival times or i.i.d. service times. To ease expression, only two classes are considered. Corresponding to the two examples are Case 5 and Case 6.

The settings of Case 5 are the same as Case 3 except that some dependence\footnote{The same series of pseudo random numbers have been used in generating the interarrival times for both classes.} is introduced in the two Poisson arrival processes. To denote this case, we use $M^*/D/1$, where $^*$ indicates that some dependence exists among classes. 

In Case 6, Class 1 has the same settings of Class 1 in Case 2, while Class 2 has the same settings of Class 2 in Case 4. In other words, the system has two classes, where one is $D/D$ and the other is $M/M$. In Figure \ref{fig:gg1}, we denote the system by $DM/DM/1$. As discussed in Section \ref{sec-21}, in this $DM/DM/1$system, customers do not identical service time distribution. 

For the two cases, the following corollaries are obtained by directly applying Theorem \ref{th-gg1}  with the characteristics of the corresponding processes. The detailed proofs are omitted. 

\begin{corollary}\label{cor:gg1-1}
For the $M^*/D/1$ example, we have for any customer $p_n^j$,  
\begin{eqnarray}
P \{ W_n^{j} > \tau \} &\lessapprox& N e^{-\theta^* \tau / N}. \label{gg1-1}
\end{eqnarray} 
with $\theta^*$ as shown in (\ref{md1-appr}).
\end{corollary}

\begin{corollary}\label{cor:gg1-2}
For the $DM/DM/1$ example, we have for any customer $p_1^j$, 
\begin{equation}
P \{ W_1^{j}   > \tau \} \le e^{- (\mu_2 - \frac{\lambda_2}{1-\rho_1}) \tau}, \label{gg1-dm-1}
\end{equation}
and for any customer $p_2^j$
\begin{equation}
P \{ W_2^{j} - Y_1  > \tau \} \le e^{- (\mu_2 - \frac{\lambda_2}{1-\rho_1}) \tau} . \label{gg1-dm}
\end{equation}
\end{corollary}

\begin{figure}[t!]
    \centering
    \subfigure[Case 5: $M^*/D/1$]{
        \includegraphics[width=0.465\linewidth]{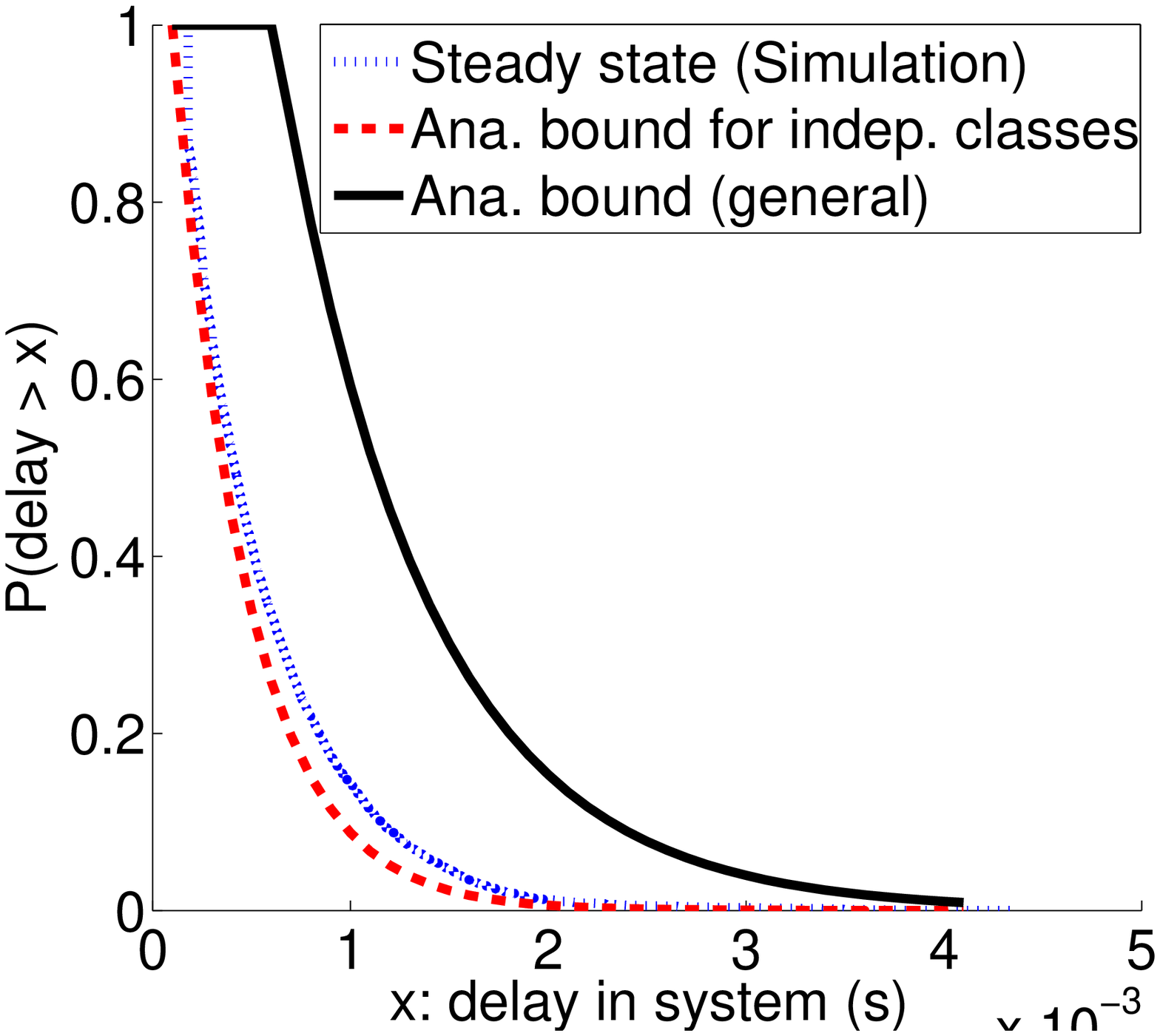} 
        \label{fig:case5}
    }
    \hfill
    \subfigure[Case 6:$DM/DM/1$] {
        \includegraphics[width=0.465\linewidth]{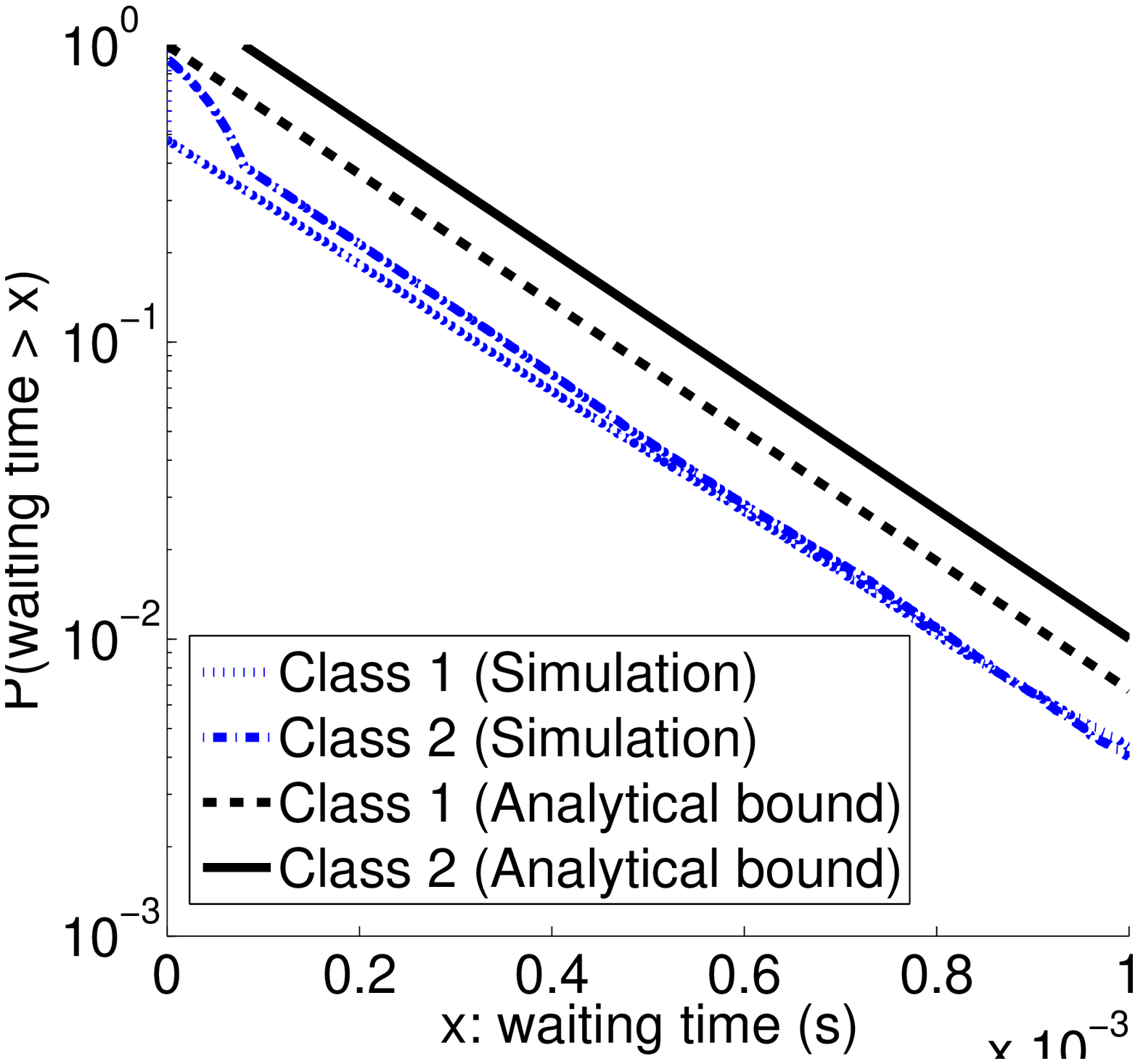} 
        \label{fig:case6}
   }
    \caption{Multiclass $G/G/1$}
    \label{fig:gg1}
\end{figure}

To illustrate the two bounds in Corollary \ref{cor:gg1-1} and Corollary \ref{cor:gg1-2}, Figure \ref{fig:gg1} presents results for the two cases, both for Class 2. Figure \ref{fig:case5} indicates that when there is dependence between the two classes, the analytical bound assuming independent classes is no more an upper-bound. However, the general analytical bound (\ref{gg1-1}) holds, even though there is a noticeable gap from the actual distribution. Note that the general bound holds for any possible dependence structure between the classes. By making use of the dependence information in the analysis, the analytical bound could be improved, but this is out of the scope of the present paper and we leave it for future investigation. 

Figure \ref{fig:case6} indicates an interesting phenomenon, which is that the (steady state) waiting time distributions of the two classes are different. Note that, here, the waiting time distribution has been intentionally used. In Case 3 - Case 5, where both classes have Poisson arrivals, the waiting time distribution is the same for both classes. However, in Case 6, while Class 2 still has Poisson arrivals, Class 1 has periodic arrivals. This arrival process difference results in the waiting time distribution difference: The waiting time observed by a Poisson inspector is different from that by a periodic inspector. Nevertheless, the bounds (\ref{gg1-dm-1}) and (\ref{gg1-dm}) are valid and are fairly good. Though improvement might be further made, the bounds provide an initial step towards the analysis of similar problems.



\appendix

\section{Proof of Lemma 1}
Our starting point is (\ref{pf1-4}). From (\ref{pf1-4}) and following the same argument as for (\ref{eq-db-4b}), the following inequality is readily obtained, which holds for all sample paths: 
\begin{eqnarray}
D^{j} &\le& \sup_{0 \le s \le a^{j}} \left[ \sum_{n=1}^{N}\frac{A_n(s, a^{j}_+)}{C_n} - (a^{j}-s) \right ] \label{eq-db-3a01} \\ 
&=&V^{j} + \sum_{n=1}^{N}\frac{A_n(a^{j}, a^{j}_+)}{C_n}
\end{eqnarray}

Define $Z(k) = e^{\theta[\sum_{n=1}^{N}\frac{A_n(a^{j}-k, a^{j})}{C_n}-k]}$, $k = 1, 2, \dots, a^{j}$, where $\theta >0$ is a constant. Then, under the condition $E[e^{\theta (\sum_{n=1}^{N}\frac{A_n(1)}{C_n} - 1)}] \le 1$, it can be proved that  $\{ Z(k) \}$ forms a supermartingale. 
\nop{
\begin{eqnarray}
&& E[Z(k+1) | Z(1), \dots, Z(k)] \nonumber \\
&=& E[e^{\theta (\sum_{n=1}^{N}\frac{A_n(a^{j}-k-1, a^{j}-k)}{C_n} - 1)}] Z(k)\nonumber \\
&=& E[e^{\theta (\sum_{n=1}^{N}\frac{A_n(1)}{C_n} - 1)}] Z(k)\nonumber \\
&\le& Z(k) \label{pr3-2}
\end{eqnarray}
and hence $\{ Z(k) \}$ forms a supermartingale. 
}
We now have, 
\begin{eqnarray}
P \{ V^{j} > \tau \}
&=& P \{ e^{\theta V^{j}} > e^{\theta \tau} \} \nonumber \\
&\le& P \{e^{\theta \sup_{s \le a^{j}} [ \sum_{n=1}^{N}\frac{A_n(s, a^{j})}{C_n} - (a^{j}-s)}] >  e^{\theta \tau}\} \nonumber \\
&=& P \{\sup_{1 \le k \le a^{j}} Z(k) >  e^{\theta \tau}\} \nonumber \\
&\le& E[Z(1)]e^{-\theta \tau} \label{pr3-3}
\end{eqnarray}
where the last step follows from the Doob's maximal inequality for supermartingale\nop{ \cite{Jiang-valuetools09} \cite{Jiang-note10}}. Since $E[Z(1)] = E[e^{\theta (\sum_{n=1}^{N}\frac{A_n(1)}{C_n} - 1)}] \equiv M_{\frac{A_n(1)}{C_n} - 1}(\theta)$, the first part is proved. 

For the second part, since $V^{j}$ and $\frac{A_n(a^{j}, a^{j}+1)}{C_n}$, $n=1, \cdots, N$, are independent, it follows from elementary probability theory results on sum of independent random variables and that $A_n(a^{j}, a^{j}_+) \le A_n(a^{j}, a^{j}+1)=_{st} A_n(1)$.

\section{Proof of Corollary 1} \label{sec:app1} 

Note that, $A_n(1)$ is a compound Poisson process with $A_n(1)=\sum_{i=1}^{\mathcal{N}_n(1)}l_n^i = \mathcal{N}_n(1) \times l_n$, where $\mathcal{N}_n(1)$ denotes the number of Class $n$ packets that arrive within a unit time.  In addition, since $C_n$ is constant, we can write $\frac{A_n(1)}{C_n}= \mathcal{N}_n(1) \times \frac{l_n}{C_n}$, which is also a compound Poisson with MGF: 
$$
E[e^{\theta \frac{A_n(1)}{C_n}}] = e^{\lambda_n (e^{\theta l_n/C_n}-1)}.
$$
Then, $M_{\sum_n\frac{A_n(1)}{C_n}-1} = e^{\sum_n \lambda_n (e^{\theta l_n/C_n}-1)-\theta}$, which implies that solving $M_{\sum_n\frac{A_n(1)}{C_n}-1} \le 1$ to get $\theta$ is equivalent to finding $\theta$ from:
\begin{equation}
\sum_n \lambda_n (e^{\theta l_n/C_n}-1) -\theta  \le 0 \label{tmp-p1}
\end{equation}
which proves the first part. 

With $e^{\theta l_n/C_n} \approx 1 +  \theta l_n/C_n + \frac{1}{2}\theta^2 (l_n/C_n)^2$ from Taylor expansion and $\theta > 0$, (\ref{tmp-p1}) can be rewritten as
$$
\theta \sum_n \lambda_n Y_n + \frac{\theta^2}{2}  \sum_n \lambda_n  Y_n^2-\theta \le 0.
$$
Since $\theta >0$, $\rho=\sum_n \lambda_n Y_n$ and $\lambda_n = X_n^{-1}$, we then get
$$
\theta \lessapprox 2(1-\rho)/(\sum_{n=1}^{N} X_n^{-1} Y_n^2). 
$$
Taking $\theta^* =2(1-\rho)/(\sum_{n=1}^{N} X_n^{-1} Y_n^2)$, the 2nd part is proved.  

\section{Proof of Corollary 2} \label{sec:app2} 

Note that, $A_n(1)$ is again a compound Poisson process with $A_n(1)=\sum_{i=1}^{\mathcal{N}_n(1)}Y_n(i)$ and similarly, $A_n(1)/C_n$ is also a compound Poisson. Since each $Y_n(i)=l_n^i/C_n$ has exponential distribution, the MGF of $A_n(1)$ can be written as, with $0 < \theta \le \min_n \mu_n$,  
$$
E[e^{\theta A_n(1) / C_n}] = e^{\frac{\lambda_n}{\mu_n - \theta} \theta}.
$$
Then, solving $M_{\sum_n\frac{A_n(1)}{C_n}-1} \le 1$ to get $\theta$ is equivalent to finding $\theta$ from:
$
\sum_n (\frac{\lambda_n}{\mu_n - \theta} \theta)-\theta  \le 0
$
and with simple manipulation, it becomes
\begin{equation}
\sum_n \frac{\lambda_n}{\mu_n - \theta} \le 1. \label{tmp-p2}
\end{equation}
which proves the first part. 

While (\ref{tmp-p2}) looks neat, finding an explicit expression for $\theta$ is not easy. In the following, we adopt an approximation approach. In particular,  
$$
\frac{\lambda_n}{\mu_n - \theta} =  \frac{\rho_n}{1 - \theta/\mu_n} \approx \rho_n (1 + \theta/\mu_n)
$$
applying which to (\ref{tmp-p2}) gives 
$
\rho + \theta \sum_n \frac{\rho_n}{\mu_n} \lessapprox 1
$
i.e., 
$$
\theta \lessapprox (1-\rho)/(\sum_{n=1}^{N} X_n^{-1} Y_n^2)
$$
since $\mu_n = Y_n^{-1}$ and $\lambda_n = X_n^{-1}$.
Then taking $\theta^* =(1-\rho)/(\sum_{n=1}^{N} X_n^{-1} Y_n^2)$, the 2nd part is proved. 

\section{Proof of Corollary 3} \label{sec:app3} 
Our starting point is (\ref{eq-db-4b}). Without loss of generality, suppose $p^j$ is a customer of class $n$. Due to also that all customers of the same class have the same service time $Y_n$ and that at time $a^{j}_{+}$, there is only one arrival that is $p^j$ and hence $W^j = D^j - Y_n$, we now have 
\begin{eqnarray}
W^{j} 
&\le& \sum_{n=1}^{N} \frac{\sup_{0 \le s \le a^{j}}[A_n(s, a^{j}_{+}) - R_n\cdot(a^{j}-s)]}{C_n}  - Y_n \nonumber \\
&=& \sum_{n=1}^{N} \frac{\sup_{0 \le s \le a^{j}}[A_n(s, a^{j}) - R_n\cdot(a^{j}-s)]}{C_n} \label{eq-db-4b-a} 
\end{eqnarray}
The right hand side of (\ref{eq-db-4b-a}) has $N$ items. Denote each as 
\begin{equation}
{\tilde{W}}_n^j = \frac{\sup_{0 \le s \le a^{j}}[A_n(s, a^{j}) - R_n\cdot(a^{j}-s)]}{C_n}. \label{eq-vw}
\end{equation}
Following the proof of the first part of Theorem 3, we have the following inequality that holds without any assumption on the potential dependence condition among classes,
\begin{equation}
P\{W^j >\tau\} \le \inf_{\sum_n p_n = 1} \sum_n P\{ {\tilde{W}}_n^j  > p_n \cdot \tau\} \label{eq-vw-2}
\end{equation}

We highlight that (\ref{eq-vw}) has a form similar to (\ref{eq-v}). Then following the same approach as for the proof of Lemma 1 and Theorem 2, we can get:
$$
P\{{\tilde{W}}_n^j \ge \tau \} \le e^{-\theta^{*}_{\omega_n} \cdot \tau}
$$
where $\omega_n \equiv \frac{R_n}{C_n}$, and $\theta^{*}_{\omega_n} $ is the solution of 
$$
E[e^{\theta(\frac{A(1)}{C_n} - \omega_n)}] = 1.
$$
For $M/D$, using similar approximation as for Corollary 1, 
\begin{equation}
\theta^{*}_{\omega_n} =2(\omega_n-\rho_n) / (X_n^{-1} Y_n^2). \label{eq-vw-theta}
\end{equation}
Finding the solution for $(\omega_1, \dots, \omega_N)$ from 
$$
2(\omega_1\_-\rho_1) / (X_1^{-1} Y_1^2) = \cdots = 2(\omega_N-\rho_N) / (X_N^{-1} Y_N^2).
$$
under the conditions $\omega_n \le \rho_n$ and $\sum_n \omega_n \le 1$, the resultant $\theta^{*}_{\omega_n}$ becomes $\theta^{*}$. 
Finally, (\ref{gg1-1}) is obtained by directly applying the resultant $P\{ {\tilde{W}}_n^j  > \tau\}$ to (\ref{eq-vw-2}). 

\section{Proof of Corollary 4} \label{sec:app4} 
If $p^j$ belongs to Class 2, we can start also from (\ref{eq-db-4b}). Following the same argument of (\ref{eq-db-4b-a}), we get
\begin{eqnarray}
W^{j} &\le& \sum_{n=1}^{N} \frac{\sup_{0 \le s \le a^{j}}[A_n(s, a^{j}) - R_n\cdot(a^{j}-s)]}{C_n}. 
\end{eqnarray}

Note that for Class 1, its customers arrive at $Y_1, 2Y_1, 3Y_1, \dots$, so, for any time period $[s, t]$, $A_1(s,t) \le r_1 \cdot (t-s) + l_1$, applying which to (\ref{eq-db-4b-a}), together with letting $R_1=r_1$, gives:
\begin{eqnarray}
W^{j} 
&\le& \frac{\sup_{0 \le s \le a^{j}}[A_2(s, a^{j}) - R_2\cdot(a^{j}-s)]}{C_2} + \frac{l_1}{C_1}. \nonumber 
\end{eqnarray}
Letting $R_2 = (1-\frac{R_1}{C_1})C_2 = (1-\rho_1) C_2$, we have
\begin{eqnarray}
W^{j} -Y_1
&\le& \sup_{0 \le s \le a^{j}}[\frac{A_2(s, a^{j})}{C_2} - (1-\rho_1)\cdot(a^{j}-s)]  \label{w-c2}
\end{eqnarray}

Following the same approach, a bound on $W^{j} -Y_1$ can be found as
$$
P\{W^{j} -Y_1 \ge \tau \} \le e^{-\theta^* \tau}
$$
where $\theta^*$ is the maximum $\theta$ satisfying
$$
E[e^{\theta A_2(1)/C_2 - (1-\rho_1)}] \le 1.
$$
Since Class 2 is $M/M$, applying it to $A_2(1)/C_2$ gives
$$
\frac{\lambda_2}{\mu_2 - \theta} - (1-\rho_1) \le 0
$$
from which, we can further get $\theta^{*} = \mu_2 - \lambda_2 / (1-\rho_1)$. 

However, if $p^j$ belongs to Class 1, we can start from (\ref{pf1-4}) and apply $A_1(s,t) \le r_1 \cdot (t-s) + l_1$ to it directly. What we then get is: 
\begin{eqnarray}
D^{j} &\le& \sum_{n=1}^{2}\frac{A_n(t^0, a^{j}_{+})}{C_n} + t^0-a^{j}  \\ 
&\le& \frac{r_1\cdot (a^{j}-t^0) + l_1}{C_1} + \frac{A_2(t^0, a^{j}_{+})}{C_2} - (a^{j}- t^0)  \nonumber  \\
&=& \frac{A_2(t^0, a^{j}_{+})}{C_2} - (1-\rho_1) (a^{j}- t^0) + \frac{l_1}{C_1} \nonumber  \\
&\le& \sup_{0 \le s \le a^{j}}[\frac{A_2(s, a^{j})}{C_2} - (1-\rho_1)\cdot(a^{j}-s)] + \frac{l_1}{C_1}.  \nonumber
\end{eqnarray}
Since $p^j$ belongs to Class 1, the above then gives 
\begin{eqnarray}
W^{j} 
&\le& \sup_{0 \le s \le a^{j}}[\frac{A_2(s, a^{j})}{C_2} - (1-\rho_1)\cdot(a^{j}-s)]  \label{w-c1}. 
\end{eqnarray}

Comparing (\ref{w-c1}) with (\ref{w-c2}), one can see the only difference is the $Y_1$ term on the left hand side of (\ref{w-c2}). Following the same approach, the waiting time distribution bound for Class 1 is obtained. 
\end{document}